\documentclass[12pt]{article}
\usepackage{rayyli} 
\newcommand{\enc}{\mathsf{Enc}}
\newcommand{\dec}{\mathsf{Dec}}

\newcommand{\zo}{\{0,1\}}

\title{Oblivious Deletion Codes}
\author{Roni Con\thanks{Department of Computer Science, Technion--Israel Institute of Technology. \href{mailto:roni.con93@gmail.com}{\texttt{roni.con93@gmail.com}}. This work was supported by the European Union (DiDAX, 101115134). 
Views and opinions expressed are those of the author(s) only and do not necessarily reflect those of the European Union or the European Research Council Executive Agency. Neither the European Union nor the granting authority can be held responsible for them. 
}
\, , Ray Li\thanks{Math \& CS Department, Santa Clara University. Email: \url{rli6@scu.edu}. Research supported by NSF grant CCF-2347371}}

\begin{document}

\maketitle

\begin{abstract}
    We construct deletion error-correcting codes in the \emph{oblivious model}, where errors are adversarial but oblivious to the encoder's randomness. 
    Oblivious errors bridge the gap between the adversarial and random error models, and are motivated by applications like DNA storage, where the noise is caused by hard-to-model physical phenomena, but not by an adversary. 
    \begin{itemize}
        \item (Explicit oblivious) We construct $t$ oblivious deletion codes, with redundancy $\sim 2t\log n$, matching the existential bound for adversarial deletions. 
        \item (List decoding implies explicit oblivious) We show that explicit list-decodable codes yield explicit oblivious deletion codes with essentially the same parameters. By a work of Guruswami and \Hastad (IEEE TIT, 2021), this gives 2 oblivious deletion codes with redundancy  $\sim 3\log n$, beating the existential redundancy for 2 adversarial deletions. 
        \item (Randomized oblivious) We give a randomized construction of oblivious codes that, with probability at least $1-2^{-n}$, produces a code correcting $t$ oblivious deletions with redundancy $\sim(t+1)\log n$, beating the existential adversarial redundancy of $\sim 2t\log n$.
        \item (Randomized adversarial) Studying the oblivious model can inform better constructions of adversarial codes. The same technique produces, with probability at least $1-2^{-n}$, a code correcting $t$ adversarial deletions with redundancy $\sim (2t+1)\log n$, nearly matching the existential redundancy of $\sim 2t\log n$.
    \end{itemize}
    The common idea behind these results is to reduce the hash size by modding by a prime chosen (randomly) from a small subset, and including a small encoding of the prime in the hash.
\end{abstract}

\newpage
\tableofcontents
\newpage

\section{Introduction}
Error-correcting codes (hereafter just codes) play a fundamental role in ensuring the reliability and integrity of data in modern communication systems. 
As information is transmitted or stored, it is often subject to noise, interference, or other disruptions that can lead to errors in the data. 
Codes add redundant information to the original data, allowing them to detect and, in many cases, correct errors without the need for retransmission. 

Tremendous progress has been made over the decades --- since the pioneering works of Shannon \cite{shannon1948mathematical} and Hamming \cite{hamming1950error} --- on the design of efficient codes that can recover from bit flips and erasures, under both probabilistic and adversarial noise models.
Deletions --- errors that remove a symbol from then transmitted word --- are another important and well-studied type of error.
Although there has been substantial advancement in recent years in understanding deletions errors, our comprehension of this model still lags behind that of codes for correcting erasures and flips. 

The main challenge in correcting deletion errors lies in the loss of synchronization between the sender and the receiver.
Consider the codeword $1001011$. If deletions occur at positions 2 and 5, the received word becomes $10111$.
Unlike the erasure setting (where the received word in this case would be $1?01?11$, explicitly indicating the missing positions), here, we \emph{do not} know where the errors occurred.
Moreover, due to the nature of deletions, most bits are shifted from their original positions, making it even harder to recover the original message.

The problem of correcting deletions has been primarily studied for two error models: adversarial and randomized. 
In the adversarial setting, the error channel is an adversary that sees the entire codeword, has unlimited computational power, and applies a fixed number of deletions.
In the randomized setting, the channel deletes every bit of the codeword independently with some fixed probability.
We refer the reader to the following excellent surveys for results in both settings \cite{Mitzenmacher08survey,mercier2010survey,cheraghchi2020overview,haeupler2021synchronization}.

In this work, we focus on the \emph{oblivious model}, where $t$ deletions are arbitrary but oblivious to the encoded codeword. Viewed another way, the deletions are applied by an adversary that knows the message but not the codeword.
For deterministic codes, the oblivious model is equivalent to the adversarial model, but if the encoding may be randomized, the oblivious model may permit less redundancy and more error tolerance. 

\begin{definition}[Oblivious Deletions Code] \label{def:stoch-codes-del}
    A stochastic code with randomized encoding function $\enc$ and (deterministic or randomized) decoding function $\dec$ is a \emph{$t$ oblivious deletion code with error $\varepsilon$} if for all messages $m$ and deletion patterns $\tau$ applying at most $t$ deletions, we have that 
    \begin{align*}
    \Pr [ \dec(\tau(\enc(m))) \neq m ] &\le \varepsilon \;,
    \end{align*} 
    The randomness of the encoder is private to the encoder and not known to the decoder.
\end{definition}

In this paper, we study oblivious deletion-correcting codes in the regime where the number of deletions, $t$, is a constant independent of the code length, $n$.
We believe that this study of oblivious deletions is important  and timely for several reasons.
\begin{enumerate}
    \item First, the oblivious model is a natural model that bridges the gap between the adversarial and random models.     
    For more well-studied types of errors like bit-flips and erasures, the oblivious model arises naturally and is well studied.
    For bit-flip and erasure errors, the oblivious channel is a specific instance of a broader class of channels known as arbitrarily varying channels (AVCs), which have been extensively studied in information theory (see the survey by Lapidoth and Narayan \cite{lapidoth1998reliable}).
    For bit-flip and erasure errors, we know their capacities for oblivious channels are $1 - H(p)$ and $1-p$ \cite{csiszar1988capacity,langberg2008oblivious}\footnote{$H(\cdot)$ denotes the binary entropy function}, matching the capacity for random channel, and there are explicit and efficient constructions \cite{guruswami2016optimal}.
    Thus, even though oblivious errors represent a stronger error model than random errors, reliable communication is still achievable at the same rate in both cases.
    Our results, along with those in \cite{guruswami2020coding}, demonstrate a similar story for oblivious deletions.
    Thus, while the errors in the oblivious model are applied adversarially, coding for oblivious errors is often more similar to coding for random errors. 
    In this sense, the model bridges the adversarial and random error models.

    \item Second, the oblivious model is well-motivated for DNA storage, one of the main motivations for deletion codes. In DNA storage systems, digital data is encoded and stored in DNA molecules sequences (with a process called synthesizing), utilizing the high density storage capacity, durability, and longevity of DNA for long-term data preservation. 
    Deletions appear naturally in DNA-based storage and thus codes that can correct them are very much needed to ensure the reliability of the stored data. The nature of errors in DNA-based storage systems is  highly dependent on elements such as the technology that was used to synthesize (write) and sequence (read) the data and the environment in which the molecules are stored. Thus, the error model, though not adversarial, is not easily predictable. 
    The oblivious model is ideally suited for this context, when the noise arises from a complex or poorly understood physical process (where the exact error distribution is unknown), but not by an adversary that sees the codeword and selects the worst-case deletions accordingly. 
    For a comprehensive survey about recent advancements regarding coding for DNA-based storage systems, we refer the reader to \cite{sabary2024survey}.

\item Third, studying the oblivious model is timely given the barriers to progress in the adversarial model. 
Designing explicit codes that can correct a constant number of adversarial deletions has received significant attention in recent years (see Table~\ref{tab:adv} for a list of works), but there appear to be barriers to progress. Adversarial $t$-deletion-codes need redundancy at least $\sim t\log n$, but existentially the best codes only achieve redundancy $\sim 2t\log n$. 
This is only beaten for the $t=1$ deletion case, where the Varshamov--Tenengolts code \cite{varshamov1965codes} achieves optimal redundancy $\sim \log n$. 
For $t\ge 2$, this factor-of-2 redundancy gap is a well-known barrier.

There is also a gap between the existential construction and the best known explicit constructions.
For $t=2$ deletions, Guruswami and \Hastad \cite{guruswami2021explicit}, improving on \cite{gabrys2018codes,sima2019two}, gives a construction matching the existential redundancy $\sim 4\log n$ and closing this gap. However, for more deletions, $t\ge 3$, a gap remains. 
Several constructions give an asymptotically optimal redundancy of $\Theta(t\log n)$, but the best construction only achieves redundancy $\sim 4t\log n$ \cite{sima2020optimal-systematic}.
Given these barriers, studying the oblivious model is relevant and timely. 

\begin{table}[t]
    \centering
    \begin{tabular}{c|c|c|c|c|c}
       Codes  & $t$ & Redundancy & Explicit? & Polytime? & Model\\ \hline
      Existential bound  & all $t$ & $\sim 2t\log n$ & No & No & Adversarial \\\hline
       VT Code \cite{varshamov1965codes}  & 1 & $\sim\log n$ & Yes & Yes & Adversarial \\\hline
    \cite{gabrys2018codes}  & 2 & $\sim 8\log n$ & Yes & Yes & Adversarial\\\hline
    \cite{sima2019two}  & 2 & $\sim 7\log n$ & Yes & Yes & Adversarial\\\hline
    \cite{guruswami2021explicit}  & 2 & $\sim 4\log n$ & Yes & Yes & Adversarial\\\hline
    \cite{belazzougui2015efficient} & all $t$ & $O(t^2+t\log n)$ & Yes & Yes & Adversarial (D.E) \\\hline
    \cite{brakensiek2017efficient}  & all $t$ & $O(t^2\log t\log n)$ & Yes & Yes& Adversarial\\\hline
    \cite{haeupler2018synchronization,cheng2022deterministic}  & all $t$ & $O(t\log^2(n/t))$ & Yes & Yes& Adversarial (D.E) \\\hline
    \cite{sima2020optimal}  & all $t$ & $\sim 8t\log n$ & Yes & Yes & Adversarial\\\hline
    \cite{sima2020optimal-systematic}  & all $t$ & $\sim 4t\log n$ & Yes & Yes & Adversarial\\\hline
    Theorem~\ref{thm:rand-adv} & all $t$ & $\sim (2t+1)\log n$ & Yes, Random & Yes & Adversarial \\\hline
       \hline
        Existential bound  & all $t$ & $\sim t\log n$ & No & No &Oblivious \\\hline
        \cite{chakraborty2016streaming}  & all $t$ & $O(t^2\log n)$ & Yes & Yes  &Oblivious (D.E)\\\hline
        \cite{belazzougui2016edit}  & all $t$ & $O(t\log^2t + t\log n)$ & Yes & Yes  &Oblivious (D.E)\\\hline
        \cite{haeupler2019optimal}  & all $t$ & $O(t\log (n/t))$ & Yes & Yes  &Oblivious (D.E)\\\hline

        Corollary~\ref{cor:list}   & 2 & $\sim 3\log n$ & Yes & Yes  &Oblivious\\\hline
        Theorem~\ref{thm:explicit}  & all $t$ & $\sim 2t\log n$ & Yes & Yes  &Oblivious\\\hline
        Theorem~\ref{thm:rand-obliv}  & all $t$ & $\sim (t+1)\log n$ & Yes, Random & Yes  &Oblivious\\\hline
       \hline
       Lower bounds & all $t$ & $\ge t\log n$ & & & Adv/Obliv
    \end{tabular}
    \caption{$t$ deletion codes under adversarial and oblivious errors. For readability, multiplicative factors of $1+o(1)$ are suppressed by the $\sim$ notation. The Polytime column means polynomial time encodable and decodable, assuming $t$ is a constant. The (D.E.) in the Model column means a result that was originally presented for the respective document exchange variant.}
    \label{tab:adv}
\end{table}

    We demonstrate that, while the oblivious model has practical relevance like the adversarial model, we can circumvent the barriers imposed by the adversarial models and construct less redundant codes (Theorem~\ref{thm:explicit}, Corollary~\ref{cor:list}, Theorem~\ref{thm:rand-obliv}). As a basic illustration of this, the above factor-of-2 gap --- between the best existential redundancy and lower bound for adversarial deletions --- does not exist for oblivious deletions (see Proposition~\ref{prop:obl-lower-bound} and Proposition~\ref{prop:exist-oblv}).
Further, as we demonstrate, studying the oblivious model can be a stepping stone to better constructions in the adversarial model. We found better constructions for oblivious deletion codes with a randomized construction (Theorem~\ref{thm:rand-obliv}), and the same techniques yield better constructions of adversarial deletion codes with randomized code constructions (Theorem~\ref{thm:rand-adv}).
\end{enumerate}

\subsection{Our results}

As far as we know, the oblivious‐deletion model has only been explicitly studied by Guruswami and Li \cite{guruswami2020coding}, who focused on the regime where the number of deletions is a constant fraction of the codeword length (another work \cite{hanna2018guess} considered a relaxed setting where the decoder succeeds with high probability over a \emph{uniformly random} codeword).
Oblivious deletions have also been studied implicitly in the form of an equivalent problem called \emph{randomized document exchange} \cite{belazzougui2016edit, chakraborty2016streaming, haeupler2019optimal} (see Related Work, Section~\ref{ssec:related} and Preliminaries, Section~\ref{ssec:equiv}).
This work considers the regime in which the number of deletions is a constant independent of the codeword length.
Our results, along with prior work, are summarized in Table~\ref{tab:adv}.

\paragraph*{Existential result and lower bound.}

We first show the optimal redundancy for $t$ oblivious deletion codes is $\sim t\log n$ by giving matching constructions and impossibility results. This stands in contrast to the adversarial model, where there is a significant gap between the best existential construction and impossibility result.
We first show the lower bound, that any code correcting $t$ oblivious deletions must have redundancy of at least $\sim t\log n$. 
We note that our lower bound is stronger than any lower bounds that may be implicit from the document exchange literature, because our lower bound holds even for non-systematic codes.

\begin{restatable}{proposition}{PropOblMinRedund} \label{prop:obl-lower-bound}
Let $n$ be a large enough integer. Let $C$ be a code with block length $n$ that can correct $t$ oblivious deletions with error $\varepsilon \leq 1/16$. Then, the redundancy of $C$ is at least 
$t\log n-O_t(\log\log n)$.
\end{restatable}

We complement this result by showing that a random code construction gives an oblivious code with optimal redundancy (up to some lower-order terms). 
\begin{restatable}{proposition}{PropOblExist}\label{prop:exist-oblv}
    Let $t$ be a constant integer and let $n$ be a large enough integer that does not depend on $t$. Let $\varepsilon\in (0,1)$. 
    Then, there exists a $t$ oblivious deletion code with error $\varepsilon$, redundancy $t\log n+ O_{t,\varepsilon}(\log\log n)$ and the amount of randomness the encoder uses is $\log\log n + O_{t,\varepsilon}(1)$.
\end{restatable}

\paragraph*{Explicit oblivious deletion codes.} Next, we turn to explicit constructions. 
\begin{restatable}{theorem}{thmexplicit}
There exists $t$ oblivious deletion code with error $\varepsilon$ and redundancy $2t\log n +O_t((\log\log n)^2 + \log\frac{1}{\varepsilon})$, encoding time $\tilde O(n)$ and decoding time $\tilde O(n^{t+1})$.
\label{thm:explicit}
\end{restatable}
\noindent 
This gives an explicit $t$ oblivious deletion code that achieves a redundancy within a factor-of-2 of optimal.
Our construction also matches (up to lower order terms) the existential redundancy $\sim 2t\log n$ of adversarial deletion codes, which is not known to be achievable by explicit adversarial codes for $t\ge 3$.

\paragraph*{List-decoding implies oblivious.} Our next result connects deletion list decodable codes with oblivious deletions, showing that an explicit deletion list-decodable code gives an explicit oblivious deletion code.
\begin{restatable}{theorem}{thmlist}
Given any code that is $(t,L)$ list-decodable against deletions with encoding and decoding times $T_{\enc}$ and $T_{\dec}$, and redundancy $r_{\textup{list}}=r_{\textup{list}}(n)$, we can construct in deterministic $\poly\log n$ time a $t$ oblivious deletion code with error $\varepsilon$ that has redundancy $r_{\textup{list}}(n) + O_t(\log(L/\varepsilon)+\log\log n)$, is encodable in time $T_{\enc}+\tilde O_t(n)$ and decodable in time $T_{\dec}+\tilde O_t(Ln)$.
\label{thm:list}
\end{restatable}

We can combine this with a result of Guruswami and \Hastad \cite{guruswami2021explicit} that gives codes with redundancy $3\log n + O(\log\log n)$ that are list-decodable against 2 deletions with list size 2.
This yields oblivious deletion codes against 2 deletions with similar redundancy.
\begin{restatable}{corollary}{corlist}
    There exists an explicit $2$ oblivious deletion code with error $\varepsilon$, redundancy $3\log n + O(\log\log n)$, and encoding and decoding times $\tilde O(n)$.
    \label{cor:list}
\end{restatable}

Observe that this $2$ oblivious deletion code has redundancy that surpasses the existential bound in the adversarial case. Namely, while the best known (non-explicit) $2$-deletion correcting code has redundancy $\sim 4\log n$, we achieve explicit and efficient $2$ oblivious deletion codes with redundancy $\sim 3\log n$.

Theorem~\ref{thm:list} further motivates studying list-decoding against deletions.
We ask the obvious follow-up question, which would imply explicit oblivious deletion codes beyond the adversarial existential bound for all constant $t$:
\begin{question*}
Are there explicit codes list-decodable against $t$ deletions with constant list size and redundancy $\sim c\log n$ for $c<2t$, or, ideally, $c=t$?\footnote{We note that in \cite[Theorem 2]{hanna2019list}, the authors construct list-decodable codes against $t$ deletions with list size $O_t(1)$ and redundancy $\sim 2t\log n$. Combining with Theorem~\ref{thm:list}, we get another construction of $t$ oblivious deletion with redundancy $\sim 2t \log n$, as in Theorem~\ref{thm:explicit}.}
\end{question*}

\paragraph*{Randomized Explicit constructions approaching the existential bound.}
Our next results benefit from the power of randomness when constructing the code. Specifically, we show that there is a randomized process that, with probability $1 - 2^{-\Omega(n)}$, produces codes that are $t$ oblivious deletion codes with redundancy $\sim (t+1)\log n$, almost matching the optimal redundancy $\sim t\log n$, and substantially beating the adversarial existential redundancy $\sim 2t\log n$.
\begin{restatable}{theorem}{thmrandobliv}
For all positive integers $t$ and $\varepsilon\in(0,1)$, for $n$ sufficiently large, there exists a randomized construction that, in time $\tilde O(n)$, with probability $1-2^{-8n}$, produces a $t$ oblivious deletion code with decoding error $\varepsilon$, redundancy $(t+1)\log n + O(\log\log n) + O(\log(1/\varepsilon))$, encoding time $\tilde O(n)$ and decoding time $\tilde O(n^{t+1})$.
\label{thm:rand-obliv}
\end{restatable}

Next, we show that the same randomized construction also gives a $t$-adversarial deletion code that almost matches the adversarial existential redundancy $\sim 2t\log n$.
\begin{restatable}{theorem}{thmrandadv}
For all positive integers $t$, and $n$ sufficiently large, there exists a randomized construction that, in time $\tilde O(n)$, with probability $1-2^{-8n}$, produces a $t$-adversarial deletion code with redundancy $\sim(2t+1)\log n$, encoding time $\tilde O(n)$ and decoding time $\tilde O(n^{t+1})$.
\label{thm:rand-adv}
\end{restatable}

\begin{remark}
    Our oblivious code constructions in Theorem~\ref{thm:explicit} and Theorem~\ref{thm:rand-obliv} enjoy two additional properties.
    First, they are \emph{systematic}, meaning that for every message $m$, the encoded codeword $c$ begins with $m$ in its first $|m|$ bits (see Definition~\ref{def:sys-code} for a formal definition). 
    Second, the decoding algorithm of these codes never outputs an incorrect message. Formally, for every message $m$ and every deletion pattern $\tau$, we have that 
    \begin{align*}
    \Pr[\dec(\tau(\enc(m)))\in\{m,\perp\}] = 1 \quad\textup{and } \quad   \Pr[\dec(\tau(\enc(m))) = \perp] \leq \varepsilon \;. 
    \end{align*}
    We observe that the codes in Theorem~\ref{thm:rand-adv} is also systematic, but the codes given in Proposition~\ref{prop:exist-oblv} and Corollary~\ref{cor:list} are not systematic.

\end{remark}

\subsection{Related work}
\label{ssec:related}
We now describe additional related works.

    \paragraph*{Document exchange.} Deletion codes are, up to lower order terms in the redundancy, equivalent to the (single-round) document exchange problem \cite{orlitsky1991interactive}.
    In the document exchange problem Alice and Bob have documents represented by strings $x$ and $x'$, respectively.
    Neither knows the others' document, but they are promised the documents are ``similar.''
    Alice's goal is to send a \emph{hash} (or \emph{summary}) to Bob so that Bob can learn $x$.
    The main question is: how long does the hash need to be (as a function of the similarity) for Bob to learn $x$ (efficiently).

    Most relevant to deletion codes is the setting where Bob's document $x'$ is a subsequence of Alice's document $x$ obtained by $t$ deletions, but the problem is studied more generally when $x$ and $x'$ are at bounded edit distance or hamming distance. 
    The edit distance case is equivalent to our setting up to constant multiplicative factors in the hash size.

    If Alice's hash is a deterministic function of her string $x$ (this is \emph{deterministic document exchange}), the optimal hash size is equal to, up to lower order terms in the redundancy, the optimal redundancy for correcting $t$ adversarial deletions (see \cite{haeupler2019optimal, cheng2022deterministic}).
    If Alice's hash is a randomized function (with private randomness) and Bob recovers with high probability (this is \emph{randomized document exchange}), the optimal hash size is equal to, up to lower order terms in the redundancy, the optimal redundancy for correcting $t$ oblivious deletions. The randomized--oblivious equivalence holds for the same reason as the deterministic--adversarial equivalence, but we justify it for completeness in Section~\ref{ssec:equiv}: in short, we can encode the hash itself in a $t$-adversarial-deletion code at negligible cost (when $t\ll n$).
    
    The document exchange problem has been the subject of extensive research \cite{barbara1988exploiting,orlitsky1991interactive,barbara1991class,abdel1994optimal,cormode2000communication,suel2004improved,irmak2005improved,jowhari2012efficient,belazzougui2015efficient,belazzougui2016edit,chakraborty2016streaming,haeupler2019optimal,cheng2022deterministic}. 
    Orlitsky showed that the hash needs to have length $\Omega(t\log (n/t))$ bits. 
    He also showed that one can compute in \emph{exponential time} a hash of $\Theta(t\log (n/t))$ bits. 
    For deterministic document exchange, Belazzougui \cite{belazzougui2015efficient} achieved a hash of size $\Theta(t^2 + t \log^2 n)$. Later, two independent works \cite{haeupler2019optimal,cheng2022deterministic} achieved a hash of size $\Theta(t\log^2(n/t))$ which is near optimal.
    The problem of designing an efficient deterministic document exchange protocol with optimal hash size $\Theta(t\log (n/t))$ remains open for $t=\Omega(n)$.
    For randomized document exchange, \cite{chakraborty2016streaming,belazzougui2016edit} achieved hash size $\Theta(t^2 \log n)$ and $\Theta(t\log^2t + t\log n)$, respectively. Later, Haeupler \cite{haeupler2019optimal} achieved optimal (randomized) hash size $\Theta(t\log(n/t))$.

    \paragraph*{Codes against oblivious errors and stochastic codes}

    Several works have studied codes that correct from oblivious substitution and erasure errors. 
    Langberg \cite{langberg2008oblivious} was the first to  explicitly define \emph{oblivious} channels, though the oblivious channel is a specific instance of a broader class of channels known as arbitrarily varying channels (AVCs), which have been extensively studied in information theory (see the survey by Lapidoth and Narayan \cite{lapidoth1998reliable}). The capacities of the oblivious substitution channel and erasure channels are $1 - H(p)$ and $1-p$, respectively \cite{csiszar1988capacity}.
    In \cite{guruswami2016optimal}, Guruswami and Smith provided the first explicit and efficient codes for the oblivious substitution channel that achieve capacity.

    We re-emphasize that codes for the oblivious errors must be stochastic to achieve better rates than for adversarial errors. 
    Stochastic codes have also found applications in the design of codes against channels that are computationally bounded. 
    Lipton \cite{lipton1994new}, Micali et al. \cite{micali2005optimal}, Chen et al. \cite{chen2015characterization}, Guruswami and Smith \cite{guruswami2016optimal}, and Shaltiel and Silbak \cite{shaltiel2021explicit-a,shaltiel2021explicit-b,shaltiel2022error,shaltiel2024explicit} studied various computationally bounded channels: some defined by polynomial‐size circuits, others by online or space‐bounded constraints. 

    \paragraph*{The binary deletion channel.}
    The binary deletion channel with parameter $p$ (BDC$_p$) is the most natural channel that models random deletions. More specifically, in this channel, each bit is deleted independently at random with probability $p$. 
    Tremendous efforts were put on determining the capacity of the BDC$_p$ and the reader is referred to the excellent surveys \cite{Mitzenmacher08survey,cheraghchi2020overview} and references therein. 

    Since our interest in this paper is on correcting a constant number of oblivious deletions, a closely related question is correcting $t$ random deletions.
    Thus, of particular interest to our question is the capacity of the binary deletion channel in the regime where the deletion probability goes to $0$. This question was addressed in two independent papers \cite{kalai2010tight,KM13} where it was shown that when $p\to 0$, the capacity of the BDC$_p$ is $1 - (1-o(1))H(p)$ (where the $o(1)$ goes to $0$ as $p\to0$). 

    \paragraph*{List decoding vs other error models.}
    Theorem~\ref{thm:list} shows a connection between list-decoding deletions and oblivious deletions and further motivates the study of list-decoding against deletions. 
    Like oblivious errors, list-decoding can be seen as a bridge between the adversarial and random error models, and exploring the connections between these various error models is a fundamental question. 
    Along these lines, one recent work \cite{pernice2024listdecoding}  answered a long-standing question by showing that list-decoding from adversarial substitutions is closely linked to random substitutions (known as the symmetric channel): any capacity-achieving list-decodable code with sufficiently large Hamming distance also achieves capacity on the symmetric channel.

\subsection{Organization}
In Section~\ref{sec:overview}, we give an overview of the main technical ideas in our constructions.
In Section~\ref{sec:prelims}, we give some preliminaries for our constructions.
In Section~\ref{sec:exist}, we show Proposition~\ref{prop:obl-lower-bound} and Proposition~\ref{prop:exist-oblv}, that the optimal redundancy for $t$ oblivious deletions is $\sim t\log n$.
In Section~\ref{sec:explicit}, we show Theorem~\ref{thm:explicit}, giving explicit $t$ oblivious deletion codes with redundancy $\sim 2t\log n$.
In Section~\ref{sec:list}, we show Theorem~\ref{thm:list}, that list-decodable deletion codes yield oblivious deletion codes.
In Section~\ref{sec:rand}, we show Theorem~\ref{thm:rand-obliv} and Theorem~\ref{thm:rand-adv}, our randomized constructions of oblivious deletion codes and adversarial deletion codes.

\section{Technical Overview}
\label{sec:overview}

Our main idea for obtaining $t$ oblivious deletion codes, which is based on \cite{sima2020optimal-systematic}, is to take $t$-adversarial-deletion hashes, modulo a random prime. 
We first illustrate this with Theorem~\ref{thm:explicit}, giving explicit $t$ oblivious deletion codes with redundancy $\sim 2t\log n$, then sketch our other results, which are variations on this idea.

By the equivalence between randomized document exchange and systematic oblivious deletion codes (Lemma~\ref{lem:equiv}), it suffices to construct randomized document exchange protocols correcting $t$ deletions with similar redundancy. Equivalently, it suffices to construct a systematic code --- a code whose encoding is $m\mapsto m\circ h(m)$ for some \emph{hash} $h(m)$ --- for the easier setting that the hash $h(m)$ is transmitted noiselessly and only the message is corrupted. 
The equivalence holds because we can encode the hash in a $t$ deletion code at negligible cost (see proof of Lemma~\ref{lem:equiv}).

Our randomized document exchange hash $h(\cdot)$ is computed as $h(m) = (h_{\textup{uniq}}(m)\mod p, p)$, where $h_{\textup{uniq}}:\{0,1\}^n\to\{0,1\}^{r_{\textup{uniq}}}$ is a deterministic document exchange hash with reasonable (but not necessarily optimal) redundancy $r=\poly\log n$, and where $p$ is a prime sampled uniformly at random from a set of $\tilde O(n^t)$ primes.
By the prime number theorem, we can sample our prime from primes $p\le \tilde O(n^t)$, so each of $(h_{\textup{uniq}}(m)\mod p)$ and $p$ can be stored in $\sim t\log n$ bits, for a total redundancy of $\sim 2t\log n$.
The decoder is the trivial decoder that, given a received word $z$ and a hash $(g,p)$, computes the hash $h_{\textup{uniq}}(m')\mod p$ of each supersequence $m'$ of the received word $z$, and outputs the (hopefully) unique supersequence $m'$ with the matching hash mod $p$.

Now we see why this is correct. For any message $m$, and any deletion pattern $\tau$, there are at most $n^t$ supersequences of the received word $\tau(m)$.
Importantly, the randomness of the hash is independent of $\tau$.
Identify the output of $h_{\textup{uniq}}$ with a number in $[2^{r_{\textup{uniq}}}]$.
Any of the $n^t$ supersequences of $\tau(m)$ have pairwise distinct hashes by the definition of deterministic document exchange.
Thus, by the Chinese Remainder Theorem, the hash $h(m)$ collides with any other hash on at most $r_{\textup{uniq}}$ primes; that is, for any $m'\neq m$, $h_{\textup{uniq}}(m)-h_{\textup{uniq}}(m')\equiv 0\mod p$ for at most $r_{\textup{uniq}}$ distinct primes. Hence, there are at most $n^t\cdot r_{\textup{uniq}}$ ``bad'' primes that cause a collision involving the hash $h_{\textup{uniq}}(m)$. 
Thus, if we sample a prime from a set of at most $n^t\cdot r_{\textup{uniq}}/\varepsilon$ primes, we avoid all bad primes with probability $1-\varepsilon$, so our decoding correctly returns the message $m$ with probability $1-\varepsilon$.
This is true for all $m$ and $\tau$, so we have our oblivious deletion code.

We can modify this for our other results as well.
\begin{itemize}
    \item We now sketch the idea of Theorem~\ref{thm:list}, which derives an oblivious deletion code from list-decodable deletion codes.
    Assume for simplicity the list-decodable deletion code is systematic, with hash $h_{\textup{list}}:\{0,1\}^n\to\{0,1\}^{r_{\textup{list}}}$ (the non-systematic case is similar). 
    Choose the hash $h(m) = (h_{\textup{list}}(m), h_{\textup{uniq}}(m)\mod p, p)$ for some $p$ sampled randomly from a set of $\frac{L}{\varepsilon}\polylog n$ primes $p$.
    Now our decoder list-decodes using $h_{\textup{list}}(m)$, then deduces the correct message from the list using $h_{\textup{uniq}}(m)\mod p$ and $p$. 
    Deducing the correct message from the list only requires distinguishing the correct message from $L$ other messages with probability $1-\varepsilon$, which requires a negligible $O(\log (L/\varepsilon) + \log\log n)$ redundant bits.

    \item To obtain better redundancy with a randomized explicit construction (Theorem~\ref{thm:rand-obliv}), we first sample a set $P$ of $O(n/\varepsilon)$ primes from the set of all primes $P_{\textup{all}}$ smaller than $\tilde O(n^t)$, where $\varepsilon$ is the decoding error.\footnote{we actually sample from primes in $[M/2,M]$, where $M=\tilde \Theta(n^t)$.}
    This set $P$ defines the code.
    Now, to encode, we sample from $P$ rather than $P_{\textup{all}}$.
    This yields redundancy $\sim(t+1)\log n$, rather than $\sim2t\log n$, because it only takes $\log n$ bits (rather than $t\log n$ bits) to index the prime.
    By concentration inequalities, for any message $m$ and deletion pattern $\tau$, with probability at least $1-2^{-10n}$, the set $P$ contains at least $1-\varepsilon$-fraction of primes that don't cause a collision.
    Union bounding over $m$ and $\tau$ gives that the resulting code is an oblivious deletion code with overwhelming probability.

    For a randomized construction of adversarial deletion codes (Theorem~\ref{thm:rand-adv}), we can use the same idea. This time, we sample a set of $O(n)$ primes from the set of all primes $P_{\textup{all}}$ that are at most $\tilde O(n^{2t})$, giving a redundancy of $\sim(2t+1)\log n$. 
\end{itemize}

\section{Preliminaries}
\label{sec:prelims}
All logs are base 2 unless otherwise specified. $\ln(x)$ is log base $e$. For an integer $k$, we denote $[k] = \{1,2, \ldots, k\}$. 
Throughout, we identify integers with their binary representations and vise versa, so that, for example, we identify $[2^a]$ with $\{0,1\}^a$, and we may speak of $x\mod p$ for a binary string $x$ and a prime $p$.
We use Landau notation $O(\cdot), \Omega(\cdot), \Theta(\cdot),o(\cdot)$.
We use the notation $\tilde O(\cdot), \tilde \Omega(\cdot), \tilde \Theta(\cdot)$ to indicate that polylog factors are suppressed, and we use the notation $O_a(\cdot), \Omega_a(\cdot), \Theta_a(\cdot)$ to indicate that multiplicative factors depending on variable $a$ are suppressed.
Throughout, the notation $\sim$ suppresses $1\pm o(1)$ factors.

A \emph{substring} of a string $s\in \zo^n$ is obtained by taking consecutive symbols from $s$ while a \emph{subsequence} of a string of $s$ is obtained by deleting some (possibly none) of the symbols in $s$. A \emph{supersequence} of a string $s$ is another string $s'$ which contain $s$ as a subsequence.
We write $x\sqsubseteq y$ to say that $x$ is a subsequence of $y$.
A \emph{run} in a string $s$ is a single-symbol substring of $s$ of maximal length. Every string can be uniquely written as the concatenation of its runs.  
For example, if $s=0111001$, then $s$ is decomposed into $0 \circ 111\circ 00 \circ 1$ where the symbol $\circ$ denotes concatenation of two strings. 
A \emph{deletion pattern} is a function $\tau$ that removes a fixed subset of symbols from strings of a fixed length.

It is well known that for any $x\in \zo^{n - t}$, the number of supersequences of length $n$ is exactly $\sum_{i=0}^{t}\binom{n}{i}$ (see, e.g., \cite[Equation 24]{levenshtein2001efficient}). For our purposes, we use the following crude upper bound.

\begin{lemma} 
    Let $n$ and $t$ be integers\footnote{The bound is slightly incorrect for $t=1$, but all our results are irrelevant anyway for $t=1$, where the VT-code \cite{varshamov1965codes} achieves optimal redundancy even for adversarial errors, so we just state and use the $t\ge 2$ version in favor of a cleaner bound.} such that $2\leq t< n/2$ and let $z\in \zo^{n-t}$. The number of strings $s\in \zo^n$ such that $z$ is a subsequence of $s$ is at most $n^t$.
    \label{lem:bound-superseq}
\end{lemma}
\begin{proof}
    We wish to show $\sum_{i=0}^{t}\binom{n}{i}\leq n^t$.
    For $t=2$, we can prove this bound directly, and for $t\ge 3$, we can use $\sum_{i=0}^{t}\binom{n}{i}\le (t+1)\binom{n}{t}\le (t+1)n^t/t! \le n^t$.
\end{proof}

\subsection{Equivalence to Randomized Document Exchange}
\label{ssec:equiv}
We show that systematic oblivious deletion codes are equivalent (up to lower order redundancy terms) to randomized document exchange hashes. First, we define a systematic code.

\begin{definition}[Systematic code]
    A code with encoding function $\Enc:\{0,1\}^n\to\{0,1\}^{n+r}$ is \emph{systematic} if, for all $x$, the first $n$ bits of $\Enc(x)$ are $x$ (with probability 1).
    \label{def:sys-code}
\end{definition}
Randomized document exchange is similar to systematic codes for oblivious deletions, but we assume the redundant symbols are uncorrupted.
\begin{definition}[Deterministic and Randomized document exchange]
    A \emph{deterministic} document exchange protocol for length $n$ messages and $t$ deletions is given by a hash function $h:\{0,1\}^n\to \{0,1\}^r$ and a decoding function $\Dec:\{0,1\}^{n-t}\times \{0,1\}^r \to \zo^n$ such that, for all $x$ and all subsequences $x'$ of length $n-t$, we can recover $x$ from $x'$ and $h(x)$. That is, $\Dec(x',h(x))=x$.\footnote{typically, the definition states that $x'$ is any string with $ED(x,x')\le t$, but this definition is equivalent up to a constant factor in $t$, and for our work this definition is more relevant.}

    A \emph{randomized} document exchange protocol with error $\varepsilon$ is a document exchange protocol such that the hash function is randomized and we recover $x$ from $x'$ and $h(x)$ with high probability: $\Pr_{h(\cdot)}[\Dec(x',h(x))=x]\ge 1-\varepsilon$.
\end{definition}
By encoding the redundant bits in a $t$-deletion code, we get that these two models are equivalent.
This equivalence arises for the same reason that deterministic document exchange is equivalent to adversarial deletion codes (see, e.g., \cite{haeupler2019optimal, cheng2022deterministic}). We note this connection is implicit in \cite{hanna2018guess}. For completeness, we state this equivalence in the next lemma and defer its proof to the appendix (Section~\ref{sec:proof-of-equiv-obl-rand-doc}).

\begin{restatable}{lemma}{EquivDocOblLem}[Randomized document exchange is equivalent to systematic oblivious deletions]
The following hold:
\begin{enumerate}
    \item Suppose we have a systematic oblivious $t$-deletion code with $n$ message bits, $r$ redundancy bits, encoding time $T_E$, and decoding time $T_D$. Then we can construct in $O(n)$ time a randomized document exchange protocol with length $n$, hash length $r$, distance $t$, error $\varepsilon$, encoding time $T_E + O(n)$, and decoding time $T_D+O(n)$.

    \item Suppose we have a randomized document exchange protocol with length $n$, hash length $r$, distance $t$, error $\varepsilon$, encoding time $T_E$, and decoding time $T_D$. Then we can construct in $O(n)$ time a systematic oblivious $t$-deletion code with $n$ message bits, $r+O(t\log(r))$ redundancy bits, encoding time $T_E + O(n+\poly r)$, and decoding time $T_D+O(n+\poly r)$.
\end{enumerate}
\label{lem:equiv}
\end{restatable}

\subsection{Concentration Inequalities}
We use the following multiplicative version of the Chernoff bound.
\begin{lemma}[Multiplicative Chernoff bound; see, e.g., \cite{mitzenmacher2017probability}]\label{lem:chernoff}
    Suppose $X_1, \ldots, X_n$ are independent identically distributed random variables taking values in $\zo$. Let $X = \sum_{i=1}^n X_i$ and $\mu = \bbE[X]$. Then, for any $0<\alpha \le 1$:
    \[
    \Pr[X > (1 + \alpha) \mu] \le e^{-\frac{\mu  \alpha^2}{2+\alpha}}
    \]
\end{lemma}
Also, we shall use Hoeffding's inequality.
\begin{lemma} [Hoeffding's inequality; see, \cite{hoeffding1994probability}] \label{lem:hoeff}
        Suppose that $X_1, \ldots, X_n$ are independent random variables with finite first and second moments and $a_i\leq X_i \leq b_i$ for $1\leq i \leq n$. Let $X = \sum_{i=1}^n X_i$ and $\mu = \bbE[X]$. Then, for any $t > 0$ we have 
    \[
    \Pr[X - \mu > t] < \exp \left( -\frac{2t^2}{\sum_{i=1}^n (b_i - a_i)^2}\right).
    \]
\end{lemma}

\subsection{Prime Number Theorem}

We utilize the prime number theorem.
\begin{theorem}[\cite{de1896recherches,hadamard1896distribution}]
Let $\pi(N)$ denote the number of primes less than $N$. Then $\lim_{n\to\infty}\frac{\pi(N)}{N/\ln N} = 1$.
In other words, for all $\varepsilon>0$, there exists $N_0$ such that, for all $N\ge N_0$,  $\pi(N)\in[(1-\varepsilon)N/\ln N, (1+\varepsilon)N/\ln N]$.
\label{thm:pnt}
\end{theorem}
We use the following corollary of the prime number theorem that follows from plugging in $\varepsilon=1/10$.
\begin{corollary}[See also Corollary 5.2 of \cite{dusart2018explicit}]
    There exists an absolute constant $N_0$ such that, for all $N\ge N_0$, the number of primes between $N/2$ and $N$ is at least $N/(10\ln N)$.
    \label{cor:pnt}
\end{corollary}
\begin{proof}
    Let $N_0$ be $\max(2N_0', 100)$, where $N_0'$ is from the prime number theorem with $\varepsilon=1/10$.
    Then the number of primes in $[N/2,N]$ is at least $0.9N/\ln N - 1.1(N/2)/\ln(N/2) > N/(10\ln N)$.
\end{proof}

\subsection{The code of \cite{belazzougui2015efficient}}

We use as a black box the following construction of deterministic document exchange protocols. We note that this construction does not achieve the state-of-the-art hash size, however, it's decoding algorithm is linear.
When the hash size of this ingredient is less important, we prefer this construction to get better decoding times.

\begin{theorem}[\cite{belazzougui2015efficient}]
    There exists a deterministic document exchange protocol for $t$ deletions with a hash of size $O(t^2+t\log^2n)$ and encoding and decoding time $\tilde O(n)$.
    \label{thm:b15}
\end{theorem}

\section{Existential Result and Lower Bound}
\label{sec:exist}
In this section, we prove Propositions~\ref{prop:obl-lower-bound} and~\ref{prop:exist-oblv}. We first introduce a definition for a \emph{deterministic} code that corrects \emph{random deletions in the average-case}.

\begin{definition}
    A deterministic code $C$ is called \emph{$t$-random deletion code (in the average-case) with error probability $\varepsilon$} if there exists a decoder $\dec :\zo^{n-t}\to C$ such that
    \[
    \Pr_{c,\tau}[\Dec(\tau(c))\neq c]\leq \varepsilon
    \]
    Here, the probability is over a uniformly random codeword and a uniformly random $t$ deletion pattern.
\end{definition}

For the sake of readability, we will write $t$-random deletion code to refer to $t$-random deletion code (in the average case).

This section is organized as follows. In Section~\ref{sec:avg-obl-equiv}, we show that if one has an oblivious deletion code, then there also exists a deterministic code that is a $t$-random deletion code. The lower bound on the redundancy of oblivious deletion codes (Proposition~\ref{prop:obl-lower-bound}) is given in Section~\ref{sec:lower-bound} with the assistance of a result by Kalai, Mitzenmacher, and Sudan \cite{kalai2010tight} that shows a lower bound on the redundancy of a $t$-random deletion code.
In Section~\ref{sec:random-strongly-obl}, we show that a random construction results in an oblivious deletion code.

\subsection{Oblivious deletions $\implies$ Random deletions in the average-case}
\label{sec:avg-obl-equiv}
In this section, we show that if we have a $t$ oblivious deletion code with error $\varepsilon$, then there exists a deterministic code $C$ that can correct $t$-random deletions with error $\Theta(\sqrt{\varepsilon})$. 
Roughly speaking, we construct $C$ by ``sampling'' the oblivious code: For each message $m$, include $c=\enc(m)$ in the codebook. 
We then show that there exists such sampling for which the resulting code is a $t$-random deletion code.

To describe the ``sampling'' process, we give a formal definition of a stochastic code (recall that oblivious codes are stochastic codes, which means that their encoder is randomized).
\begin{definition}[stochastic code]
    A \emph{stochastic binary code} with redundancy $r\in \bbN$, \emph{randomness length} $b\in \bbN$ and \emph{block length} $n \in \bbN$ is given by an encoder $\enc: \zo^{n-r} \times \zo^b \rightarrow \zo^n$ and a decoder $\dec: \zo^*\rightarrow \zo^{n-r} \cup \{\perp\}$.
\end{definition}

Note that the encoder function $\enc$ takes a message and a random string. However, we mostly write $\enc(v)$ to refer to the process of sampling a random $u\leftarrow \zo^b$ and then encoding the message $v$ with $\enc(v,u)$. Nevertheless, there are some places where we explicitly write $\enc(v,u)$ to indicate an encoding with a specific seed.

\begin{claim} \label{clm:obl-to-random}
    Let $t$ be a constant integer and $n$ be a large enough integer. Assume that $D$ is a stochastic code of length $n$, redundancy $r$, and assume that $D$ is a $t$ oblivious deletion code with error $\varepsilon$. 
    Then, there exists a deterministic code $C$ that can correct from $t$-random deletions with the same length and redundancy, and with error $\leq 2\sqrt{\varepsilon}$.
\end{claim}

\begin{proof}
    Let $D$ be an $(\enc,\dec)$ stochastic code with redundancy $r$, length $n$. Define $C = \{ \enc(v) \mid v\in \zo^m\}$ to be a deterministic code obtained by sampling $2^m$ codewords via the encoder of $D$. Since $C$ is a deterministic code, we shall abuse notation and refer to $C$ as also being the encoding map from the space of messages to the codewords.
    
    We shall define a new decoder $\dec':\zo^{n-t} \to \zo^n$ such that \[\dec'(\tau(c)) = C(\dec(\tau(c)))\;.\] 
    Namely, the decoder of our deterministic code applies the decoder of the oblivious code to get a message and then encodes it using our sampled code $C$.

    Enumerate all the message by $v_1,\ldots,v_{2^m}$ and their encoded codewords by $c_1, \ldots, c_{2^m}\in C$. For every deletion pattern $\tau$, define $I_{\tau} = |\{i : \dec'(\tau(c_i))\neq c_i\}|$ to be the number of codewords that are incorrectly decoded by the decoder of our deterministic sampled code $C$. 
    For every $\tau$, we have that 
    \begin{align*}
        \E_{c_1,\ldots,c_{2^m}}[I_\tau] 
        &=\sum_{i=1}^{2^m} \Pr_{c_1,\cdots,c_{2^m}}[\dec'(\tau(c_i)) \neq c_i]
        \nonumber\\
        &=\sum_{i=1}^{2^m} \Pr_{\enc(v_1),\ldots,\enc(v_{2^m})}[ C(\dec(\tau(\enc(v_i)))) \neq C(v_i)]
        \nonumber\\
        &=\sum_{i=1}^{2^m} \Pr_{\enc(v_1),\ldots,\enc(v_{2^m})} [\dec(\tau(\enc(v_i))))\notin \{v_j | C(v_j) = C(v_i)\}] \\
        &\leq \sum_{i=1}^{2^m} \Pr_{\enc(v_1),\ldots,\enc(v_{2^m})}[\dec(\tau(\enc(v_i)))\neq v_i] \\
        & = \sum_{i=1}^{2^m} \Pr_{\enc(v_i)}[\dec(\tau(\enc(v_i)))\neq v_i] \\
        &\le \sum_{i=1}^{2^m} \varepsilon \nonumber\\
        &= \varepsilon \cdot 2^m\;,
    \end{align*}
    where the third equality follows since $C$ might not be injective and therefore there can be a set of multiple messages that give the correct codeword. The first inequality follows by noting that $v_i$ clearly belongs to that set. The fourth equality follows by noting that, according to Definition~\ref{def:stoch-codes-del}, the probability that the decoder fails is over the randomness of the encoding process of the specific message.
    
    Thus, summing upon all deletion patterns, we get $\E_{c_1,\ldots,c_{2^m}}[\sum_{\tau} I_{\tau}] \leq \binom{n}{t} \varepsilon 2^m$, and thus, there exists a choice of $c_1,\ldots,c_{2^m}$ for which $\sum_{\tau}I_\tau \leq \binom{n}{t}\cdot \varepsilon 2^m$. 
    Fix $c_1,\ldots, c_{2^m}$ to be such a choice. 
    It must be that for at least $(1 - \sqrt{\varepsilon})$ fraction of the deletion patterns, we have that $I_\tau \leq \sqrt{\varepsilon} 2^m$. Indeed, otherwise, $\sum_{\tau : I_\tau > \sqrt{\varepsilon}} I_\tau > \sqrt{\varepsilon} \binom{n}{t} \cdot \sqrt{\varepsilon} 2^m = \binom{n}{t}\varepsilon 2^m$, which is a contradiction. 
    Now, we compute the probability that upon a random codeword and a random deletion pattern, the decoding fails. 
    We have 
    \begin{align*}
        \Pr_{i,\tau} [\dec'(\tau(c_i))\neq c_i] &= \Pr_{\tau}[I_\tau >\sqrt{\varepsilon}2^m] \cdot \Pr_{i}[\dec'(\tau(c_i))\neq c_i|I_\tau > \sqrt{\varepsilon} 2^m] \\
        &\;\;+ \Pr_{\tau}[I_\tau \leq \sqrt{\varepsilon}2^m] \cdot \Pr_{i}[\dec'(\tau(c_i))\neq c_i|I_\tau \leq \sqrt{\varepsilon} 2^m] \\
        & \leq \sqrt{\varepsilon} + (1-\sqrt{\varepsilon})\cdot \sqrt{\varepsilon} \\
        &\leq 2\sqrt{\varepsilon} \;.
    \end{align*}
\end{proof}

\subsection{Lower bound on the redundancy of oblivious deletion codes}
\label{sec:lower-bound}
In \cite{kalai2010tight}, Kalai, Mitzenmacher, and Sudan studied the setting where the channel is the binary deletion channel with deletion probability $p$ (BDC$_p$). In this case, every bit is deleted independently, with probability $p$.
They showed that if $C$ is a code that is robust against the BDC$_p$, in the regime where $p$ is small, then the rate of $C$ is at most $1 - (1 - o(1))H(p)$ (where the $o(1)$ goes to $0$ when $p\to 0$).

However, to prove their result, they first proved a slightly weaker theorem which contains two relaxations. First, they relax the channel to the case where there are \emph{exactly} $pn$ random deletions. Second, the decoding algorithm succeeds on a \emph{random codeword} and not on every codeword. 
The statement of the theorem, rephrased to our setting, i.e., to the case where the number of deletions is constant (independent of the block length), is given next.

\begin{restatable}{theorem}{KMSTheorem}\cite[Theorem 2.2, rephrased]{kalai2010tight} \label{thm:kms-lower-bound}
    Let $n$ be a large enough integer. Let $C$ be a code with block length $n$ that is a $t$-random deletion code in the average case with error probability $\varepsilon$.
    Then, the redundancy of $C$ is at least 
    \[
    \log \binom{n}{t} + t - \log(3t) -\log(2/(1-\varepsilon)) -O(t \log \log(n))\;.
    \]
\end{restatable}
\begin{remark}
    We remark that our statement is slightly different from \cite[Theorem 2.2]{kalai2010tight}. Thus, for completeness, we provide a proof in the appendix.
\end{remark}

We are now ready to prove Proposition~\ref{prop:obl-lower-bound} which is restated next.
\PropOblMinRedund*
\begin{proof}
    Let $C$ be a code with block length $n$ that can correct $t$ oblivious deletions with error $\varepsilon$. 
    By Claim~\ref{clm:obl-to-random}, we get that there exists a code $C'$ with the same redundancy and block length that can correct $t$ random deletions in the average case with error $2\sqrt{\varepsilon}$. 
    Thus, by Theorem~\ref{thm:kms-lower-bound}, the redundancy of $C$ (and $C'$) is at least $\log\binom{n}{t} + t -\log 3t-\log(2/(1-2\sqrt{\varepsilon})) -O(t\log\log n)$. The desired result is obtained by recalling that $\varepsilon \leq 1/16$.
\end{proof}
\subsection{Random construction}
\label{sec:random-strongly-obl}

In this section, we prove Proposition~\ref{prop:exist-oblv} which is stated here again, for convenience
\PropOblExist*

\begin{proof}
    Set $s=\varepsilon^{-2} \cdot 10t\cdot \log n$ and
    consider $|C|$ messages $m_1,\dots,m_{|C|}$ where $|C|= \frac{\varepsilon}{2s}\cdot \frac{2^n}{n^t}$.
    For each message $m_i$, we associate a multi-set of codewords $E_i \in \zo^n$ of size $s$ where every element in $E_i$ is a uniform, random vector in $\zo^n$. 
    The encoder and the decoder of our oblivious code are defined as follows
    \begin{itemize}
        \item $\enc:\zo^{\log|C|}\to \zo^n$. A message $m_i$ is mapped to a random element of $E_i$.
        \item $\dec: \zo^{n-t}\to \zo^{\log|C|} \cup\{\perp\}$. For a received word $z$ of length $n-t$, the decoder finds all messages $m_i$ where $z$ is a subsequence of a possible encoding of $m_i$ ($z\sqsubset w$ for some $w\in E_i$). 
        If there is only one such message, the decoder returns that message. 
        Otherwise, the decoder return $\perp$.
    \end{itemize}
    A string $w\in E_i$ is called \emph{$(\tau, i)$ bad}, if there exists $u\in E_j$ where $j\neq i$ such that $\tau(w) \sqsubseteq u$. 
    By Lemma~\ref{lem:bound-superseq}, the number of strings in $\zo^n$ that contain $z$ as a subsequence is at most $n^t$, and thus, the probability that there exists $w'$ in one of the $E_j$s such that $\tau(w)\sqsubseteq w'$ is at most $|C| \cdot s \cdot \frac{n^t}{2^n} \leq \varepsilon/2$.
    
    We say a message $m_i$ is \emph{$\tau$-bad} if $|\{w\in E_i \mid w \textup{ is }(\tau, i) \textup{ bad}\}|>\varepsilon \cdot s$.
    Denote by $X^{\tau}_w$ the random variable indicating that $w$ is $(\tau, i)$ bad. Observe that for $w,w'\in E_i$,
    $X^{\tau}_w$ and $X^{\tau}_{w'}$ are independent events due to the process that was used to choose the $E_i$s. 
    Thus, we can apply the Hoeffding's inequality when bounding the probability that a fixed message $m_i$ is $\tau$-bad for a given $\tau$.
    \begin{align*}
    \Pr[m_i \text{ is }\tau\text{-bad}] &= \Pr_{E_i}\left[\big|\{w\in E_i \mid \tau(w) \text{ is bad}\}\big| \geq \varepsilon s\right] \\
    &= \Pr_{E_i}\left[\sum_{w\in E_i} X^{\tau}_w \geq (1 + 1)\frac{\varepsilon s}{2}\right]\\
    &\leq e^{- 2\varepsilon^2 s/4}\\
    &\leq n^{-5t}
    \end{align*}
    
    Let $X_\tau$ be the number of messages that are $\tau$-bad and observe that  $\E[X_\tau] \le |C|/n^{5t}$. Thus, by linearity of expectation, we get that  $\E[\sum_{\tau} X_{\tau}] = \binom{n}{t} \cdot  |C|/n^{5t}\leq |C|/n^{4t}$. Therefore, there exists some choice of randomness for which the sets $E_1,\ldots E_{|C|}$ are such that $\sum_{\tau} X_\tau \le |C|/n^{4t}$. 
    For this choice of sets, we define the set of messages to be all the messages that are \emph{not} $\tau$ bad for every $\tau$. 
    The number of such messages is at least $|C| - |C|/n^{4t} \geq |C|/2$, for large enough $n$.
    Therefore, the redundancy of our code is at most 
    \[
    \log (n^t) + \log 2 s - \log \varepsilon + 1 \leq t\log n + \log\log n + \log t - 3\log \varepsilon + O(1) \;,
    \]
    
    Now, we show that the probability of error is at most $\varepsilon$. In fact, we show something stronger: that the decoder never outputs a wrong message and that the probability that it outputs $\perp$ is $\leq \varepsilon$. Indeed, first observe that by the definition of the decoder, upon an input $\tau(\enc(m_i))\in \zo^{n-t}$, if there are $w\in E_i$ and $w'\in E_j$ such that $s\sqsubseteq w$ and $s\sqsubseteq w'$, then it outputs $\perp$. Since we always have that $\tau(\enc(m_i)) \sqsubset w$ for some $w\in E_i$, we conclude that the decoder never outputs a wrong message.
    Second, for every message $m_i$ and every $t$-deletion pattern $\tau$, by our construction guarantees above, we have that $|\{w\in E_i \mid w \textup{ is } (\tau, i) \textup{ bad} \}| \leq \varepsilon |E_i|$. The claim follows by recalling that our encoder selects a uniform string in $E_i$.
\end{proof}

\section{Explicit oblivious deletion codes with redundancy $\sim 2t\log n$}
\label{sec:explicit}

We now prove Theorem~\ref{thm:explicit}, which gives explicit oblivious deletion codes with redundancy $\sim 2t\log n$. The following lemma is the core of the construction, showing how to turn a "reasonable" deterministic document exchange protocol into a randomized document exchange protocol with redundancy $\sim 2t\log n$.
\begin{lemma}
  Suppose there exists hashes of length $f(n,t)$ for deterministic document-exchange for $t$ deletions on words of length $n$ with encoding time $E(n,t)$. 
  Then there exist hashes of length $r=2\log \binom{n}{t} + \log f(n,t) + \log\frac{100}{\varepsilon}$ for randomized document exchange with error $\varepsilon$.
  Furthermore, the encoding takes $E(n,t)+\polylog n$ time and the decoding takes $E(n,t)\cdot O(n^t)$ time.
  \label{lem:explicit}
\end{lemma}
Applying this with the construction in \cite{belazzougui2015efficient}, which achieves $f(n,t) = O(t^2+t\log^2n)$ and $E(n,t) = \tilde O(n)$, and then using the equivalence between randomized document exchange and oblivious deletions (Lemma~\ref{lem:equiv}), we obtain Theorem~\ref{thm:explicit}.
\thmexplicit*
\begin{proof}[Proof of Lemma~\ref{lem:explicit}]
  Let $\varepsilon$ be the desired error.  
  Let $M=100n^t\cdot \frac{f(n,t)}{\varepsilon}$ and $r=2\log M = 2t\log n + O_t(\log\log n)$.
  Let $h_{\textup{uniq}}:\{0,1\}^n\to [2^{f(n,t)}]$ be the hash for deterministic document exchange.

  \paragraph*{Encoding.}
  Define the hash $h:\{0,1\}^n\to \{0,1\}^{r}$ for our randomized document exchange protocol as follows:
  \begin{itemize}
    \item The encoding chooses a random prime number $p$ between $M/2$ and $M$.
    The hash is $(h_{uniq}(m)\mod p, p)$. 
  \end{itemize}

  \paragraph*{Decoding.}
  Suppose that $z$ is the received word of length $n- t$ (if it is a longer subsequence, apply additional deletions arbitrarily) and $(g,p)$ is the hash.
  \begin{itemize}
    \item By brute force search, find all supersequences $m'$ of $z$ such that $h_{\textup{uniq}}(m')\mod p = g$. If there is exactly one such $m'$, return $m'$, else return $\perp$.
  \end{itemize}

  \paragraph*{Correctness.}
  Fix a message $m$ of length $n$ and a subsequence $z$ of length $n-t$.
  Let $(h_{\textup{uniq}}(m)\mod p,p)$ be the hash of $m$. 
  We show that with probability $1-\varepsilon$ over the randomness of the encoder, there are no supersequences $m'$ of $z$ such that $h_{\textup{uniq}}(m') \equiv h_{\textup{uniq}}(m)\mod p$.
  Indeed, since $h_{\textup{uniq}}(\cdot)$ is a hash for deterministic document exchange, then $h_{\textup{uniq}}(m')\neq h_{\textup{uniq}}(m)$ for all supersequences $m'$ of $z$ with $m\neq m'$.
  We also have $|h_{\textup{uniq}}(m)-h_{\textup{uniq}}(m')|\le 2^{f(n,t)}$, which means this difference has at most $\log_{M/2}2^{f(n,t)} \le \frac{2f(n,t)}{\log M}$ prime factors greater than $M/2$.
  Since there are at least $\frac{M}{10\log M}$ prime factors between $M/2$ and $M$ by the Prime Number Theorem (Corollary~\ref{cor:pnt}), the probability that $h_{\textup{uniq}}(m)\equiv h_{\textup{uniq}}(m')\mod p$ is at most $\frac{2f(n,t)/\log M}{M/(10\log M)} = \frac{20f(n,t)}{M}$.
  By a union bound over the $n^t$ possible values of $m'$ (Lemma~\ref{lem:bound-superseq}), the probability that $m$ is decoded incorrectly is at most $n^t\cdot\frac{20f(n,t)}{M} < \varepsilon$.

  \paragraph*{Runtime.}
  The encoding takes time $E(n,t)$, plus the time to generate the prime $p$, which takes $\poly\log n$ if we generate the $O(M)\le \tilde O(n^t)$ primes in advance.
  The decoding is dominated by the brute force search, which searches $\binom{n}{t}$ strings and checking each one takes $E(n,t)$ time.
\end{proof}

\section{List decoding implies oblivious}
\label{sec:list}
In this section, we show that an explicit and efficient list-decodable deletion code yields an explicit and efficient oblivious code with effectively the same redundancy. We restate Theorem~\ref{thm:list} for convenience
\thmlist*

\begin{proof}
    We describe the encoding and decoding algorithms, then justify the correctness and runtime.
    \paragraph*{Encoding.}
    Let $\Enc_{\textup{list}}:\{0,1\}^n\to \{0,1\}^{n+r_{\textup{list}}}$ be the encoding of the $(t,L)$-list-decodable code.
    Let $h_{\textup{uniq}}:\{0,1\}^{n+r_{\textup{list}}}\to [2^{\alpha t\log^2 2n}]$ be the hash of the deterministic document exchange protocol in Theorem~\ref{thm:b15} on messages of length $n+r_{\textup{list}}$, where $\alpha>0$ is an absolute constant.
    Let $\textup{Rep}_t(x)=x_1^tx_2^t\cdots x_{|x|}^t$ denote the string $x$ where each bit is repeated $t$ times.
    Note that $\textup{Rep}_t$ encodes a code that corrects $t-1$ deletions.
    
    Let $P$ be the set of primes in $[M/2,M]$, where $M= 100\cdot \alpha t\cdot \frac{L}{\varepsilon}\cdot \log^2 2n$. 
    By the Prime Number Theorem (Corollary~\ref{cor:pnt}), $P$ has at least $\frac{M}{10\ln M}$ primes.
    Define a new encoding function $\enc':\{0,1\}^n\to\{0,1\}^{n'}$, for $n'=n+r_{\textup{list}} + r_{\textup{rep}}$ for $r_{\textup{rep}}\defeq (t+1)\ceil{\log |P|+\log M}$ where $\enc'$ chooses a uniformly random prime $p$ from $P$, and then sets
    \begin{align}
        \enc'(m) = \enc_{\textup{list}}(m)\circ \textup{Rep}_{t+1}(\ab{p, h_{\textup{uniq}}(\enc_{\textup{list}}(m))\mod p}).
    \end{align}
    where $\ab{p, h_{\textup{uniq}}(m)\mod p}$ denotes the binary representation of $(p, h_{\textup{uniq}}(m)\mod p)$.
    Since $p$ can be represented in $\ceil{\log|P|}$ bits and $h_{\textup{uniq}}(m)\mod p$ can be represented in $\ceil{\log M}$ bits, the redundancy of this encoding is at most $r_{\textup{list}}+r_{\textup{rep}} \le r_{\textup{list}} + O(t\log tL/\varepsilon + t\log\log n)$ 

    \paragraph*{Decoding.}
    Suppose we are given a received word $z$.
    Let $z_0$ be the first $n+r_{\textup{list}}-t$ bits of $z$, and let $z_1$ denote the last $r_{\textup{rep}}-t$ bits of $z$. 
    Run the list-decoding algorithm of $C$ on $z_0$ to compute a list $\mathcal{L}$ of $L$ messages. 
    Run the repetition code decoder on $z_1$, and let the decoded word be $\ab{p,g}$ for some prime $p$ and hash $g$.
    Then iterate through the list $\mathcal{L}$ to find messages $m$ such that $h_{\textup{uniq}}(m)\mod p = g$. If there is a unique $m$, return that $m$, else return $\perp$.
    
    \paragraph*{Correctness.} Fix a message $m$ and a deletion pattern $\tau$, let $z=\tau(\enc'(m))$ be the received word, and let $z_0,z_1$ be the substrings of $z$ in the decoding algorithm.
    By construction, $z_0$ is a subsequence of $\enc_{\textup{list}}(m)$ obtained by applying $t$ deletions.
    Hence, the list-decoding algorithm correctly determines a list $\mathcal{L}$ of $L$ messages such that, (i) for each message $m'$, the string $z_0$ is a subsequence of $\enc(m')$, and (ii) the correct message $m$ is in the list.
    Similarly, $z_1$ is a subsequence of $\textup{Rep}_{t+1}(\ab{p, h_{\textup{uniq}}(\enc_{\textup{list}}(m))\mod p})$ obtained by applying $t$ deletions, so the repetition decoder correctly determines $p$ and $g=h_{\textup{uniq}}(\enc_{\textup{list}}(m))\mod p$.
    For any $m'\neq m$ in the list $\mathcal{L}$, we must have $h_{\textup{uniq}}(\enc_{\textup{list}}(m)) \neq h_{\textup{uniq}}(\enc_{\textup{list}}(m'))$ as $\enc_{\textup{list}}(m)$ and $\enc_{\textup{list}}(m')$ are distinct length-$(|z_0|+t)$ supersequences of $z_0$. Further, their difference satisfies $|h_{\textup{uniq}}(\enc_{\textup{list}}(m)) - h_{\textup{uniq}}(\enc_{\textup{list}}(m'))| \le 2^{\alpha t\log^2 2n}$.
    This means they share at most $\log_{M/2} 2^{\alpha t\log^2 2n} \le \frac{\alpha t\log^2 2n}{\log(M/2)} < \frac{\varepsilon}{L}\cdot \frac{M}{10\ln M} \le \frac{\varepsilon}{L}\cdot |P|$ common prime factors in $P$.
    The string $z_0$ is independent of the randomness of the encoding, so the list $\mathcal{L}$ is independent of the randomness of the encoding.
    Hence, a random prime from $P$ fails to distinguish $m$ and $m'$ with probability at most $\frac{\varepsilon}{L}$, so the probability a random prime fails to distinguish $m$ from all other list codewords $m'\in \mathcal{L}$ is at most $\varepsilon$ by the union bound.
    Thus, we fail to recover $m$ uniquely with probability at most $\varepsilon$, as desired.

    \paragraph*{Runtime.} The construction takes time $\polylog n$ to compute the list of primes $P$. The encoding time is the time $T_{enc}$ to encode in the list-decoding hash, plus the time $\tilde O(n)$ to encode in the document exchange hash of Theorem~\ref{thm:b15}, plus the time $O(M)$ to choose a prime and compute the mod, for a total time of $T_{enc} +\tilde O(tn)$.
    The decoding time is the time $T_{dec}$ to list-decode, plus the iteration to filter the list $O(Ltn)+\tilde O(L\cdot n)$, where we need to check for each $m$ in the list, whether $h_{\textup{uniq}}(\enc_{\textup{list}}(m))\mod p = g$. 
\end{proof}

Combining this with the 2-deletion codes of Guruswami and \Hastad \cite{guruswami2021explicit} gives oblivious 2-deletion codes with redundancy $\sim 3\log n$, which is Corollary~\ref{cor:list}.  
\begin{theorem}[\cite{guruswami2021explicit}]
    There exist explicit codes encodable and decodable in linear time $O(n)$ with redundancy $3\log n +O(\log\log n)$ that are list-decodable against 2 deletions with list size 2.
\end{theorem}

\corlist*

\section{Randomized explicit oblivious and adversarial codes approaching the existential bound}
\label{sec:rand}

\subsection{Randomized explicit: Oblivious with $\sim(t+1)\log n$ redundancy}

We now prove Theorem~\ref{thm:rand-obliv}, which gives a randomized construction of $t$ oblivious deletion codes approaching the optimal lower bound and beating the adversarial existential bound.

\thmrandobliv*

\begin{proof}
    By Lemma~\ref{lem:equiv}, it suffices to give a protocol for randomized document exchange with the provided redundancy, encoding time, and decoding time.
    We describe the construction, encoding and decoding algorithms, then justify the correctness and runtime.

\paragraph*{Construction.}
Let $P_{\textup{all}}$ be the set of all the primes in $[M/2,M]$ for $M=100M_0\ln M_0$ for $M_0=4\alpha\varepsilon^{-1}n^t\log n$, where $\alpha$ is the absolute constant from Theorem~\ref{thm:b15}.
Randomly choose a multiset $P$ of $100n/\varepsilon$ primes chosen independently at random from $P_{\textup{all}}$.
This set $P$ specifies the code.

\paragraph*{Encoding.}
Let $h_{\textup{uniq}}:\{0,1\}^n\to\{0,1\}^r$ be a hash for deterministic document exchange for $t$ deletions, where $r=\alpha t\log^2 n$ and $\alpha$ is an absolute constant given by Theorem~\ref{thm:b15} (we can omit the $t^2$ additive factor since we assume $n$ is sufficiently large relative to $t$).
Our randomized document exchange hash is $h(m,p) \defeq (h_{\textup{uniq}}(m)\mod p,p)$ for some $p$ chosen uniformly at random from $P$.
Since we know $p\in P$, we can store a description of $p$ in $\log |P|=O(\log n) + O(\log(1/\varepsilon))$ bits, so we can store $h(m,p)$ in $\log M + \log|P| \le (t+1)\log n + O(\log\log n) + O(\log(1/\varepsilon))$ bits.

\paragraph*{Decoding.}
Suppose $z$ is a length $n-t$ received word, and suppose $(g,p)$ is the hash. 
By brute force search, find all supersequences $m'$ of $z$ such that $h_{\textup{uniq}}(m') \mod p = g$. If there is exactly one such $m'$, return $m'$. Otherwise, return $\perp$.

\paragraph*{Correctness.}
By the Prime Number Theorem (Corollary~\ref{cor:pnt}), $P_{\textup{all}}$ has at least $M/10\ln M > M_0$ primes.

Fix $\tau$. We consider the probability (over the choice of $P$) that some message $m\in\{0,1\}^n$ is decoded incorrectly with probability at least $\varepsilon$ (over the randomness of the encoder).

For any $m$, call a prime $p$ \emph{good} (with respect to $m,\tau$) if our decoder recovers message $m$ given $\tau(m)$ and $h(m,p)$ and \emph{bad} otherwise. 
Since $h_{\textup{uniq}}(\cdot)$ is a deterministic document exchange hash, for any of the at-most-$n^t$ messages $m'$ that have $\tau(m)$ as a subsequence, we have $h_{\textup{uniq}}(m)\neq h_{\textup{uniq}}(m')$. 
Since $|h_{\textup{uniq}}(m)-h_{\textup{uniq}}(m') |\le 2^{\alpha t\log^2n}$, there are at most $\log_{M/2} 2^{\alpha t\log^2n} \le \alpha \log n$ primes in $P_{\textup{all}}$ dividing $h_{\textup{uniq}}(m)-h_{\textup{uniq}}(m')$. 
By a union bound over the at-most-$n^t$ supersequences of $\tau(m)$, at  most $\frac{n^t\cdot\alpha\log n}{|P_{\textup{all}}|} < \frac{\varepsilon}{2}$ fraction of the primes in $P_{\textup{all}}$ are bad (with respect to $m,\tau$).
Thus, the number of bad primes in $P$ is distributed as a binomial distribution over $|P|$ elements with mean $\mu < \varepsilon |P|/2$.
By the Chernoff bound (Lemma~\ref{lem:chernoff} with $\alpha=(\varepsilon |P| - \mu)/\mu\ge 1$), the probability that more than $\varepsilon$ fraction of the primes in $P$ are bad is at most $2^{-\alpha^2 \mu/(2+\alpha)} \le 2^{-(\varepsilon|P|/2)\cdot \alpha/(2+\alpha)}\le 2^{-\varepsilon|P|/6}\le 2^{-10n}$ where we used $\alpha\mu \ge \varepsilon|P|/2$ and $\alpha\ge 1$. 
Thus, with probability at least $1-2^{-10n}$ over the choice of $P$, for a fixed deletion pattern $\tau$ and fixed $m$, string $m$ is decoded correctly from $\tau(m)$ with probability at least $1-\varepsilon$.
By a union bound over all messages $m$ and all deletion patters $\tau$, with probability at least $1-2^n\cdot \binom{n}{t}\cdot2^{-10n} \geq 1 - 2^{-8n}$, for every $m$ and $\tau$, string $m$ is decoded correctly from $\tau(m)$ with probability at least $1-\varepsilon$.

\paragraph*{Runtime.}
We can construct the code (find $P$) with rejection sampling.
The density of primes is at least $\Omega(\frac{1}{\log n})$ by the Prime Number Theorem (Theorem~\ref{thm:pnt}), so rejection sampling finds enough primes in time $\tilde O(n)$ with probability at least $1-2^{-\Omega(n)}$.
Encoding takes time $O(\log n)$ to choose the prime $P$ and $\tilde O(n)$ to encode into $h(\cdot,\cdot)$.
Decoding takes time $O(n^{t+1})$ to brute force search over $m$ and hash each one with $h_{\textup{uniq}}$.
\end{proof}

\subsection{Randomized Explicit: Adversarial with $\sim(2t+1)\log n$ redundancy}

\thmrandadv*

\begin{proof}
By the equivalence between adversarial deletion codes and deterministic document exchange (see e.g., \cite{cheng2022deterministic,haeupler2019optimal}), it suffices to give a deterministic document exchange protocol with the provided redundancy. 
We describe the encoding and decoding algorithms, then justify the correctness and runtime.

Call strings $m$ and $m'$ \emph{confusable} if they share a length $n-t$ subsequence. We have the following straightforward lemma
\begin{lemma}
    Fix $m\in \zo^n$. The number of $m'\in \zo^n$ that are confusable with $m$ is at most $n^{2t}$. \label{lem:confusable}   
\end{lemma}
\begin{proof}
    The number of subsequences of $m$ of length $n-t$ is at most $\binom{n}{t}\leq n^t$ (choosing the $t$ deleted symbols). Each one of these subsequences is contained in at most $n^t$ strings in $\zo^n$ (Lemma~\ref{lem:bound-superseq}).
\end{proof}
Let $h_{\textup{uniq}}:\{0,1\}^n\to\{0,1\}^r$ be a hash for deterministic document exchange for $t$ deletions, where $r=\alpha t\log^2 n$ for some absolute constant $\alpha>0$ (Theorem~\ref{thm:b15}).
Call a message/prime pair $(m,p)$ \emph{good} --- where $m$ is a length $n$ string and $p$ is a prime --- if, for all strings $m'\neq m$ confusable with $m$, we have $h_{\textup{uniq}}(m)\mod p\neq h_{\textup{uniq}}(m')\mod p$.

\paragraph*{Construction.}
Let $P_{\textup{all}}$ be the set of primes in $[M/2,M]$ for $M=100M_0\ln M_0$ for $M_0=2\alpha n^{2t}\log n$, where $\alpha$ is the absolute constant from Theorem~\ref{thm:b15}. 
Randomly choose a multiset $P$ of $10n$ primes chosen independently at random from the primes $P_{\textup{all}}$.
This set $P$ specifies the code.

\paragraph*{Encoding.}
Our deterministic document exchange hash is $h_2(m) = (h_{\textup{uniq}}(m)\mod p,p)$ for some $p$ chosen from $P$ so that $(m,p)$ is a good message/prime pair.
We compute $p$ by brute force search.
Since we know $p\in P$, we can store a description of $p$ in $\log |P|=O(\log n)$ bits, so we can store $h_2(m)$ in $\log|P| + \log M \le (2t+1)\log n + O(\log\log n)$ bits.

\paragraph*{Decoding.}
Suppose $z$ is a length $n-t$ received string, and suppose $(g,p)$ is the hash. 
By brute force search, find all supersequences $m'$ of $z$ such that $h_{\textup{uniq}}(m') \mod p = g$. If there is exactly one such $m'$, return $m'$, else return $\perp$.

\paragraph*{Correctness.}
By the Prime Number Theorem (Corollary~\ref{cor:pnt}), $P_{\textup{all}}$ has at least $M/10\ln M > M_0$ primes.

We consider the probability (over the choice of $P$) that some message $m\in\{0,1\}^n$ is decoded incorrectly. 
Since $h_{\textup{uniq}}(\cdot)$ is a deterministic document exchange hash, for any of the at-most $n^{2t}$ many messages $m'$ that are confusable with $m$, we have $h_{\textup{uniq}}(m)\neq h_{\textup{uniq}}(m')$.
Since $|h_{\textup{uniq}}(m)-h_{\textup{uniq}}(m')|\le 2^{\alpha t\log^2n}$, there are at most $\log_{M/2} 2^{\alpha t\log^2n} \le \alpha \log n$ primes in $P_{\textup{all}}$ dividing $h_{\textup{uniq}}(m)-h_{\textup{uniq}}(m')$. Thus, for at most $n^{2t}\cdot\alpha\log n$ primes $p\in P_{\textup{all}}$, the pair $(m,p)$ fails to be good. Thus, at least $1-\frac{n^{2t}\cdot\alpha\log n}{|P_{\textup{all}}|} > 1/2 $ of the primes in $P_{\textup{all}}$ form a good pair with $m$. 
Thus, the probability that $P$ fails to contain a good prime for $m$ is at most $2^{-|P|}\le 2^{-10n}$.
Thus, with probability at least $1-2^{-10n}$ over the choice of $P$, a fixed message $m$ is decoded correctly.
By a union bound over the $2^n$ possible messages $m$, with probability at least $1-2^{-9n}$, every message $m$ is decoded correctly.

\paragraph*{Runtime.}
We can construct the code (find $P$) with rejection sampling.
The density of primes is at least $\Omega(\frac{1}{\log n})$ by the Prime Number Theorem, so rejection sampling finds enough primes in time $\tilde O(n)$ with probability at least $1-2^{-\Omega(n)}$.
Encoding is dominated by the time it takes to find a good prime for the message $m$. 
Indeed, we compute the $n^{2t}$ hashes $h_{\textup{uniq}}(m')$ of all messages $m'$ confusable with $m$, each of which takes time $\tilde O(n)$, and then take all hashes modulo $p$ for all $10n$ primes in $P$. Overall, the encoding is done in time $\tilde O(n^{2t+1})$.
Decoding takes time $O(n^{t+1})$ to brute force search all the supersequences $m'$ of $\tau(m)$ and hash each one with $h_{\textup{uniq}}$.
\end{proof}

\bibliographystyle{alpha}
\bibliography{refs}

\section{Appendix}

\subsection{Proof of Lemma~\ref{lem:equiv}}
\label{sec:proof-of-equiv-obl-rand-doc}
To prove Lemma~\ref{lem:equiv}, we use (in a black-box fashion) the following deterministic document exchange protocol, which achieves the optimal hash size and polynomial decoding time.

\begin{theorem}[\cite{cheng2022deterministic}]
    There exists a deterministic document exchange protocol for $t$ deletions with hash size $O(t\log n)$ and encoding and decoding time $\poly n$, where the exponent in the polynomial is independent of $t$.
    There also exists a systematic adversarial $t$-deletion code with redundancy $O(t\log n)$ and encoding and decoding time $\poly n$, where the exponent in the polynomial is independent of $t$.
    \label{thm:cjlw22}
\end{theorem}

We now prove Lemma~\ref{lem:equiv} which is restated for convenience.
\EquivDocOblLem*
\begin{proof}
    For the first part, suppose we have a systematic oblivious $t$-deletion code with error $\varepsilon$ with encoding $\Enc(x) = x\circ h(x)$ for some hash $h:\{0,1\}^n\to \{0,1\}^r$ and decoder $\Dec:\{0,1\}^{n+r-t}\to \{0,1\}^n$.
    Let $\Dec':\{0,1\}^{n-t}\times \{0,1\}^r$ be the decoder $\Dec'(x,h) = \Dec(x\circ h)$.
    By definition, $(h,\Dec')$ is a randomized document exchange protocol with length $n$, distance $t$, hash length $t$, and error $\varepsilon$.

    Now suppose we have a randomized document exchange protocol with length $n$, hash $h:\{0,1\}^n\to\{0,1\}^r$ of length $r$, distance $t$, error $\varepsilon$, and decoder $\Dec_h:\{0,1\}^{n-t}\times \{0,1\}^r\to\{0,1\}^n$.
  Let $\Enc_0:\{0,1\}^{r}\to \{0,1\}^{r+g(r,t)}$ and $\Dec_0:\{0,1\}^{r+g(r,t)-t}\to \{0,1\}^{r}$ be the encoding and decoding functions, respectively, for the deletion code given by Theorem~\ref{thm:cjlw22}, where $g(r,t)\le O(t\log r)$ and the encoding and decoding functions running in $\poly r$ time (the exponent in the polynomial is independent of $t$).
  We define our oblivious deletion code's encoding and decoding functions as follows.

  \paragraph*{Encoding.}
  Let the encoding for our $t$ oblivious deletion code $\Enc:\{0,1\}^n\to \{0,1\}^{n+r+g(r,t)}$ be defined as $\Enc(m) = m\circ \Enc_0(h(m))$.

  \paragraph*{Decoding.}
  Suppose that $z$ is a subsequence of $m\circ\Enc_0(h)$ of length $n+r+g(r,t) - t$ (if it is a longer subsequence, apply $t$ deletions arbitrarily).
  \begin{itemize}
    \item Let $z_0$ be the first $n-t$ bits of $z$, and let $z_1$ be the last $r+g(r,t)-t$ bits of $z$.
    \item Compute $\Dec_h(z_0,\Dec_0(z_1))$. 
  \end{itemize}

  \paragraph*{Correctness.}
  Fix a deletion pattern $\tau$ and a message $m$.
  Again, let $z_0$ be the first $n-t$ bits of $z$, and let $z_1$ be the last $r+g(r,t)-t$ bits of $z$.
  It's easy to see that $z_0$ is a subsequence of $m$ and $z_1$ is a subsequence of $\Enc_0(h(m))$.
  Over the randomness of the encoder, $m$ is fixed, so $z_0$ is fixed for a fixed $\tau$.
  Since $(\Enc_0,\Dec_0)$ corrects $t$ adversarial deletions, we know that the $\Dec_0(z_1)=h(m)$.
  Further, since $(h,\Dec_h)$ gives a randomized document exchange protocol, $\Dec_h(z_0, h(m))$ returns $m$ with probability at least $1-\varepsilon$, as desired.

  \paragraph*{Runtime.} Encoder $\Enc_0$ and decoder $\Dec_0$ runs in $\poly r$ time (because we chose the code from \cite{cheng2022deterministic}), hash $h$ runs in $T_E$ time, and $\Dec_h$ runs in $T_D$ time, so the total encoding time is $T_E+O(n)+\poly r$ and decoding time is $T_D + O(n)+\poly r$.
\end{proof}

\subsection{Proof of Theorem~\ref{thm:kms-lower-bound}}
As discussed above, Theorem~\ref{thm:kms-lower-bound} is merely a reformulation of \cite[Theorem 2.2]{kalai2010tight}. Thus, for completeness, we reproduce their original proof here, altering only the parameters slightly.

Here, we will define a deletion pattern $\tau \in [n]$ that represents $t$ deletions to be the set of indices that are \emph{not} deleted. 
The set of all $t$-deletion patterns is denoted by $P_{t,n} = \{ \{a_1, a_2, \ldots, a_{n-t}\}\mid a_1 <a_2 <\cdots < a_{n-t}\}$.
For a string $x\in \zo^{n}$, $\tau(x)$ denotes $x$ after applying the deletion pattern to it. We abuse notation and denote by $\tau$ the set of indices that are not affected by the deletion pattern and also the function that applies the deletions.

The first definition presents a \emph{distance} function between two deletion patterns. 
\begin{definition}
    Let $\tau = \{a_1,\ldots,a_{n-t}\}$ and $\tau'=\{b_1,\ldots,b_{n-t}\}$ be two $t$-deletion patterns. Denote by $\Delta (\tau, \tau')$, the number of disagreements between $\tau$ and $\tau'$:
    \[
    \Delta(\tau,\tau') = |\{i \mid a_i\neq b_i\}|
    \]
\end{definition}

Next, we give the definition of an \emph{$(\ell, t)$-bad string}. This is a string in $\zo^n$ for which there are two ``$\ell$-far apart'' deletion patterns $\tau,\tau'$ such that $\tau(x)=\tau'(x)$. Formally,  
\begin{definition}
    Let $\ell \geq 1$. A string $x\in \zo^{n}$ is $(\ell,t)$-bad if there are two distinct $t$-deletion patterns such that $\tau(x) = \tau'(x)$ and $\Delta(\tau, \tau') \geq \ell$.
\end{definition}
In \cite{kalai2010tight}, the authors showed that the number of $(\ell,t)$ bad strings is not big. Specifically, 
\begin{lemma}\cite[Lemma 2.3]{kalai2010tight}
    Let $\ell \geq 1$. There are at most $\binom{n}{t}^2 2^{n-\ell}$ different $(\ell, t)$-bad strings of length $n$.
\end{lemma}

An important step in their proof was to bound for a given $t$-deletion pattern $\tau$, how many other $t$-deletion patterns are $\ell$-close to it. Formally,
\begin{lemma}\cite[Lemma 2.2]{kalai2010tight} \label{lem:close-distance-patterns}
    Let $\ell>0$ be an integer. For any $t$-deletion pattern $\tau$, the number of $t$-deletion pattern $\tau'$ such that $\Delta (\tau, \tau')\leq \ell$ is at most 
    \[
    (\ell + 1) \binom{2t + \ell + 1}{2t+1} \binom{t+\ell}{t}\;.
    \]
\end{lemma}

Another useful lemma that \cite{kalai2010tight} used was the following
\begin{lemma} \label{lem:helping-joint-prob}
    Let $\rho$ be a joint distribution over $S\times T$ for finite sets $S$ and $T$ such that the marginal distribution over $S$ is uniform. Let $g:T \to S$ be a function. Then, $\Pr_{(a,b)\sim \rho}[g(b) = a]\leq |T|/|S|$.
\end{lemma}

We are now ready to present the proof of Theorem~\ref{thm:kms-lower-bound}. The theorem is restated for convenience.

\KMSTheorem*

\begin{proof}
    Set $\ell = 3t\log n$. The number of $(\ell,t)$-bad strings is at most
    \[
    \binom{n}{t}^2\cdot 2^{n-3t\log n}\leq \left(\frac{n^t}{t!}\right)^{2} \cdot \left( \frac{1}{n}\right)^{3t} \cdot 2^n = \frac{1}{(t!)^2}\cdot \frac{2^n}{n^t}\;,
    \]
    where the inequality follows by using the upper bound $\binom{n}{t}\leq n^t/t!$.
    Let $C$ be the code guaranteed by the theorem, and let $\mathcal{A}$ be the decoding algorithm of the code.
    
    We will next define an algorithm $\mathcal{G}$ that gets as input $\tau(c)$ for some random codeword $c$ and random $t$-deletion pattern $\tau$ and outputs $(c, \tau)$ with nonnegligible probability. On input $s$, $\mathcal{G}$ runs the decoder $\mathcal{A}$ on $s$ and then returns $\mathcal{G}(s) = (\mathcal{A}(s), \tau')$ where $\tau'$ is the lexicographically first deletion pattern for which $s = \tau'(\mathcal{A}(c))$. Now, the probability that the decoder $\mathcal{A}$ succeeds and that the codeword is \emph{not} an $(\ell,t)$-bad string is at least
    \[
    (1 - \delta) - \frac{2^n}{(t!)^2 \cdot n^t} \cdot \frac{1}{|C|} \geq \frac{1-\delta}{2}\;.
    \]
    Indeed, at the worst-case scenario, all the $(\ell,t)$-bad strings are codewords.
    The inequality above holds since otherwise we would have that $|C|\leq \frac{2^n}{(t!)^2 \cdot n^t}\cdot \frac{2}{1-\delta}$ which implies that the redundancy of $C'$ is at least $t\log n + 2\log (t!) -\log(2/(1-\delta)) > \log\binom{n}{t} + t -\log(3t) - \log(2/(1-\delta))$ for all $t$ and so the claim holds.
    
    Fix $\tau$. Conditioned on $c$ not being $(\ell,t)$-bad, for all $\tau'$ such that $\tau'(c) = \tau(c)$, we have that $\Delta(\tau, \tau')\leq \ell-1$. By Lemma~\ref{lem:close-distance-patterns}, the number of such close patterns is at most
    \begin{align*}
    \ell \cdot \binom{2t+\ell}{2t+1}\cdot \binom{t + \ell -1}{t} 
    &\leq \ell \cdot \left(e\cdot\frac{2t + \ell}{2t+1}\right)^{2t+1}\left(e\cdot\frac{t+\ell}{t}\right)^t 
    \leq \ell \cdot \left(\frac{6\ell}{t}\right)^{3t+1} \;,
    \end{align*}
    where the first inequality follows by the upper bound $\binom{n}{t}\leq (en/t)^t$ which holds for $t<n/2$ and by noting that $\ell > 2t$.
    Now, since each deletion pattern is equally likely, the probability that $\mathcal{G}(\tau(c)) = (c, \tau)$ is at least $(1-\delta)/ (2\alpha)$ where $\alpha = \ell \cdot \left(\frac{6\ell}{t}\right)^{3t+1}$.

    Now, using Lemma~\ref{lem:helping-joint-prob} with the sets $S=C\times P_{t,n}$ and $T=\zo^{n-t}$, we get the probability that $g(\tau(c)) = (c, \tau)$ is at most $\frac{2^{n-t}}{|C|\cdot \binom{n}{t}}$. Therefore, we have
    \[
    \frac{2^{n-t}}{|C| \cdot \binom{n}{t}} \geq \frac{1-\delta}{2\alpha}
    \]
    and by taking logarithm on both sides and rearranging terms, we find that
    \[
    \log(|C|) \leq n-t -\log \binom{n}{t} -\log((1-\delta)/2) + \log\alpha \;
    \]
    and the claim follows by noting that $\log(\alpha) = \log(3t\log n) + (3t+1)\log(18 \log n) = O(t\log \log n)$.
\end{proof}

\end{document}